\documentclass{article} 
\usepackage[final]{colm2025_conference}

\usepackage{microtype}
\usepackage{hyperref}
\usepackage{url}
\usepackage{booktabs}
\usepackage{tabularx}
\usepackage{lineno}

\usepackage{amsthm}

\usepackage{graphicx}
\usepackage{subcaption}
\usepackage{float}
\usepackage{multirow}

\usepackage{amsthm}  
\usepackage{ragged2e} 

\usepackage{newtxtext,newtxmath}

\newtheorem{theorem}{Theorem}[section] 
\newtheorem{proposition}[theorem]{Proposition} 
\newtheorem{lemma}[theorem]{Lemma}

\title{Generative AI as a Non-Convex Supply Shock: Market Bifurcation and Welfare Analysis}

\author{Yukun Zhang \\
The Chinese University Of Hongkong\\
HongKong, China \\
\texttt{215010026@link.cuhk.edu.cn} \\
\And
Tianyang Zhang \\
University of Bologna \\
Bologna, Italy \\
\texttt{tianyang.zhang@studio.unibo.it} \\
}

\begin{document}

\ifcolmsubmission
\linenumbers
\fi

\maketitle
\begin{abstract}
The diffusion of Generative AI (GenAI) constitutes a supply shock of a fundamentally different nature: while marginal production costs approach zero, content generation creates congestion externalities through information pollution. We develop a three-layer general equilibrium framework to study how this non-convex technology reshapes market structure, transition dynamics, and social welfare. In a static vertical differentiation model, we show that the GenAI cost shock induces a kinked production frontier that bifurcates the market into exit, AI, and human segments, generating a ``middle-class hollow'' in the quality distribution. To analyze adjustment paths, we embed this structure in a mean-field evolutionary system and a calibrated agent-based model with bounded rationality. The transition to the AI-integrated equilibrium is non-monotonic: rather than smooth diffusion, the economy experiences a temporary ecological collapse driven by search frictions and delayed skill adaptation, followed by selective recovery. Survival depends on asymmetric skill reconfiguration, whereby humans retreat from technical execution toward semantic creativity. Finally, we show that the welfare impact of AI adoption is highly sensitive to pollution intensity: low congestion yields monotonic welfare gains, whereas high pollution produces an inverted-U relationship in which further AI expansion reduces total welfare. These results imply that laissez-faire adoption can be inefficient and that optimal governance must shift from input regulation toward output-side congestion management.
\end{abstract}



\section{Introduction} \label{sec:intro}

The emergence of Generative Artificial Intelligence (GenAI) represents a supply shock of historical magnitude, comparable to the decoupling of physical labor from muscle power during the Industrial Revolution. By reducing the marginal cost of semantic generation to near-zero levels, Large Language Models (LLMs) promise a democratization of creativity and a surge in aggregate productivity \citep{brynjolfsson2025generative}. However, the prevailing techno-optimist narrative often neglects a critical general equilibrium effect: the congestion externality. In a digital economy where human attention is the scarce constraint, an exponential increase in synthetic content does not merely reduce prices; it fundamentally alters the signal-to-noise ratio, creating what we term "Information Pollution."

The central tension of the GenAI era is therefore not scarcity, but discoverability. As the cost of producing "plausible but mediocre" content vanishes, platforms face a flood of synthetic noise that degrades the matching efficiency between consumers and high-quality creators. Existing literature has largely analyzed these shifts in isolation—focusing either on the labor displacement effects \citep{acemoglu2020robots} or the algorithmic dynamics of recommendation systems \citep{lian2021optimal}. There remains a lack of unified theoretical frameworks that rigorously model how these forces interact dynamically to reshape market structures and social welfare.

To bridge this gap, this paper develops a Three-Layer General Equilibrium Framework to analyze the displacement effects and transitional dynamics of the GenAI supply shock. Unlike standard models that assume convex production frontiers and perfect information, our framework explicitly incorporates a non-convex technological choice and an endogenous pollution penalty.

Methodologically, we proceed in three steps. First, we construct a static model of vertical differentiation to characterize the asymptotic equilibrium. Second, to understand the path dependence, we model the skill evolution of creators as a Mean Field Game governed by a Fokker-Planck equation. Third, to capture heterogeneity and bounded rationality, we instantiate these dynamics in a large-scale Agent-Based Model (ABM) calibrated to the "Goldilocks" regime of late 2025—a state where human survival is possible but precarious.

Our analysis yields three distinct, counter-intuitive findings:

The "Middle-Class Hollow": We demonstrate that GenAI acts as a non-convex technology that does not simply compete with humans but bifurcates the market. The equilibrium is characterized by a "barbell" structure: low-end demand is completely captured by AI anchors at marginal cost, while human creators are forced into a high-premium, high-complexity niche. The mid-tier of the skill distribution—the traditional "middle class" of the creative economy—evaporates.

The "Shock Therapy" Transition: Contrary to the smooth S-curve adoption models typically assumed in innovation diffusion theory, our dynamic simulations reveal a non-monotonic transition. The market experiences a temporary "Ecological Collapse"—a valley of death where the income shock from AI adoption outpaces the rate of human skill reconfiguration. This suggests that without intervention, the transition to an AI-integrated economy involves significant frictional costs and a temporary breakdown in supply-side coordination.

Asymmetric Skill Reconfiguration: We identify the specific mechanism of human survival. In the "Polluted Equilibrium," survival depends on an orthogonal shift in capabilities. Agents who attempt to compete on technical efficiency (syntax, structure) are eliminated, while those who pivot toward semantic creativity (nuance, intent) survive. This validates the hypothesis that AI commoditizes execution while increasing the premium on intent.

Finally, we discuss the normative implications of these findings. We argue that the governance of GenAI must shift its focus from input-side regulation (copyright and data ownership) to output-side congestion management. We derive an optimal Pigouvian tax on algorithmic volume that internalizes the pollution externality, restoring Pareto efficiency in a noise-saturated market.

The remainder of this paper is organized as follows. Section \ref{sec:theory} derives the theoretical conditions for market bifurcation. Section \ref{sec:setup} details the "Goldilocks" ABM environment. Section \ref{sec:results} presents the simulation results regarding the "Shock Therapy" phenomenon. Section \ref{sec:discussion} outlines the policy shift from intellectual property to information pollution control.

\begin{figure*}[htbp]
    \centering
    \includegraphics[width=\textwidth]{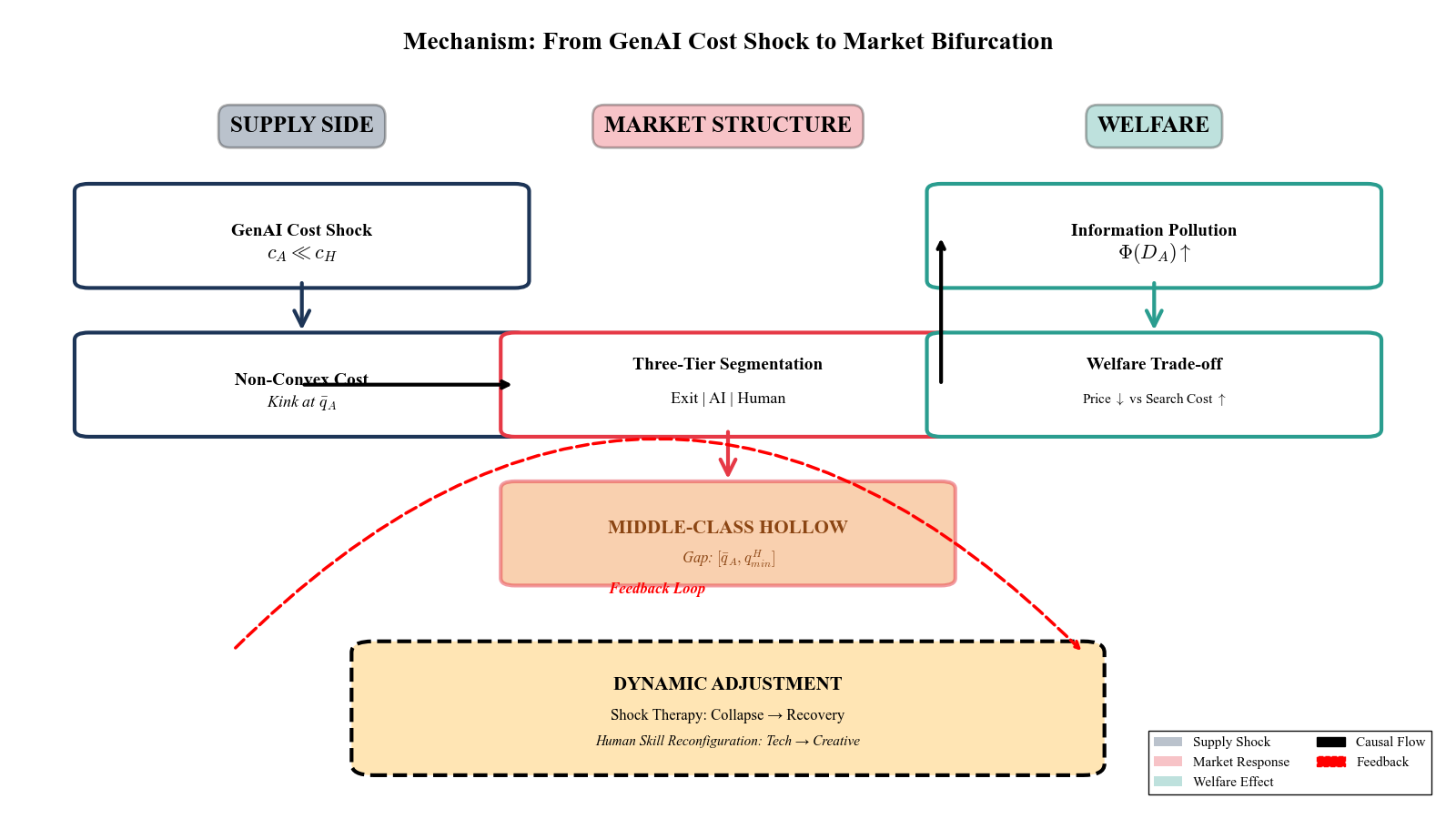}
    \caption{
    \textbf{Mechanism: From GenAI cost shock to market bifurcation and welfare trade-offs.}
    A GenAI marginal-cost shock ($c_A \ll c_H$) induces a non-convex production frontier (kink at the AI quality cap $\bar{q}_A$), generating three-tier segmentation (Exit / AI / Human) and a ``middle-class hollow'' gap in the quality distribution. The expansion of AI volume increases information pollution $\Phi(D_A)$, creating a welfare trade-off between lower prices and higher search costs. These forces feed back into the transition dynamics, producing ``shock therapy'' (collapse $\rightarrow$ recovery) and asymmetric skill reconfiguration (Tech $\rightarrow$ Creative).
    }
    \label{fig:mechanism_overview}
\end{figure*}

\section{Literature Review}

The integration of Generative AI into the digital economy fundamentally alters platform mechanics, necessitating a re-examination of established economic theories. This review synthesizes three distinct strands of literature—platform economics, information dynamics, and social welfare—to contextualize the unique supply shocks and congestion externalities introduced by GenAI.

\subsection{Generative AI in Two-Sided Markets} Classical two-sided market theory focuses on how platforms internalize cross-group externalities through optimal pricing structures and design choices \citep{rysmann2009economics}. While foundational, this literature has evolved to address friction: recent studies extend these models to incorporate congestion effects, strategic pricing behaviors, and dynamic equilibrium adjustments under capacity constraints \citep{bernstein2021competition, lian2021optimal}.

Generative AI represents a structural break in these dynamics. On the supply side, it precipitates a collapse in marginal production costs, drastically lowering entry barriers for content creation \citep{varian2018artificial}. On the demand side, it reshapes consumption through hyper-personalized algorithmic curation, theoretically amplifying engagement \citep{werthner2024introduction, hassan2025moderating}. However, while these shifts strengthen network effects, they simultaneously risk market saturation and quality dilution. Current literature lacks a unified equilibrium framework that explicitly models how this specific combination—near-zero production costs coupled with algorithmic intermediation—impacts platform stability.

\subsection{Information Overload and Distribution Dynamics} The sheer velocity of AI-generated content exacerbates the tension between digital abundance and human cognitive constraints. A substantial body of work documents the adverse effects of information overload, linking it to cognitive fatigue, reduced decision quality, and user attrition \citep{Eppler01112004, bawden2009dark, WANG2025114436}. Although personalized algorithms are deployed to mitigate this friction, evidence suggests they may inadvertently concentrate exposure on a narrow subset of popular items, reinforcing "winner-takes-all" dynamics \citep{viswanathan2017dynamics, ding2025unveiling}.

Conversely, proponents argue that AI-driven recommendation systems can revitalize the "Long Tail" by reducing discovery costs for niche content \citep{nderson2008chris, OLMEDILLA2019113120}. This unresolved duality—between algorithmic amplification of popular content and the potential democratization of niche visibility—remains a critical open question. Most existing studies examine these phenomena in isolation, failing to model the joint dynamics where supply-side pollution and demand-side filtering interact to determine market concentration.

\subsection{Welfare and Social Implications} The welfare implications of Generative AI are characterized by a sharp trade-off between aggregate efficiency and distributional equity. On one hand, GenAI significantly reduces search frictions and enhances productivity, thereby increasing consumer surplus and accelerating the diffusion of innovation \citep{goldfarb2019digital, brynjolfsson2025generative, cockburn2022impact}. On the other hand, these efficiency gains may come at the cost of intensified market concentration and the entrenchment of incumbent advantages, raising concerns about labor displacement and inequality \citep{furman2019ai, acemoglu2020robots}.

Further complexities arise from the externalities of synthesis, including copyright dilution \citep{gaffar2025copyright} and algorithmic biases in visibility allocation \citep{binns2018s}. While scholars agree on the necessity of regulatory frameworks that balance innovation incentives with fairness \citep{agrawal2019economic}, the current discourse remains largely descriptive. There is a marked absence of micro-founded, quantitative frameworks capable of rigorously assessing how specific platform governance mechanisms can internalize these new forms of externalities in the GenAI era.

\section{Theoretical Framework: The Analytical Anchor}
\label{sec:theory}

To analyze the general equilibrium effects of Generative AI, we construct a vertical differentiation model augmented with negative network externalities. Unlike standard models \citep{mussa1978monopoly} where supply is convex and information is perfect, we introduce two structural modifications: (1) an endogenous ``information pollution'' term in the consumer utility function, and (2) a non-convex production frontier representing the discrete technological choice between human labor and algorithmic generation.

\subsection{Demand Side: Endogenous Information Pollution}

Consider a continuum of consumers with measure normalized to 1, indexed by their taste for quality $\theta \sim U[0, \bar{\theta}]$. Consumers face a choice between purchasing a digital good of quality $q$ at price $p$ or selecting an outside option with utility normalized to zero.

We depart from canonical models by positing that the proliferation of synthetic content creates a congestion externality (e.g., search costs, cognitive load, or signal-to-noise degradation). Let $D_A$ denote the aggregate volume of AI-generated content in the market equilibrium. The indirect utility of a consumer of type $\theta$ is given by:

\begin{equation} \label{eq:utility}
    U(\theta, q, p, D_A) = \theta q - p - \underbrace{\beta \ln(1 + \eta D_A)}_{\Phi(D_A)}
\end{equation}

where $\Phi(D_A)$ represents the information pollution penalty. $\beta \ge 0$ measures the consumer's sensitivity to market noise, and $\eta > 0$ is a scaling parameter. Crucially, $D_A$ is an endogenous equilibrium object, defined as the total measure of firms adopting AI technology: $D_A = \int_{s \in \mathcal{S}_A} d\mu(s)$.

Let $\hat{\theta}$ denote the marginal consumer indifferent between purchasing and the outside option. The participation constraint implies $\hat{\theta} q - p - \Phi(D_A) = 0$. This yields a participation threshold $\hat{\theta} = \frac{p + \Phi(D_A)}{q}$.

\begin{figure*}[htbp]
    \centering
    \includegraphics[width=\textwidth]{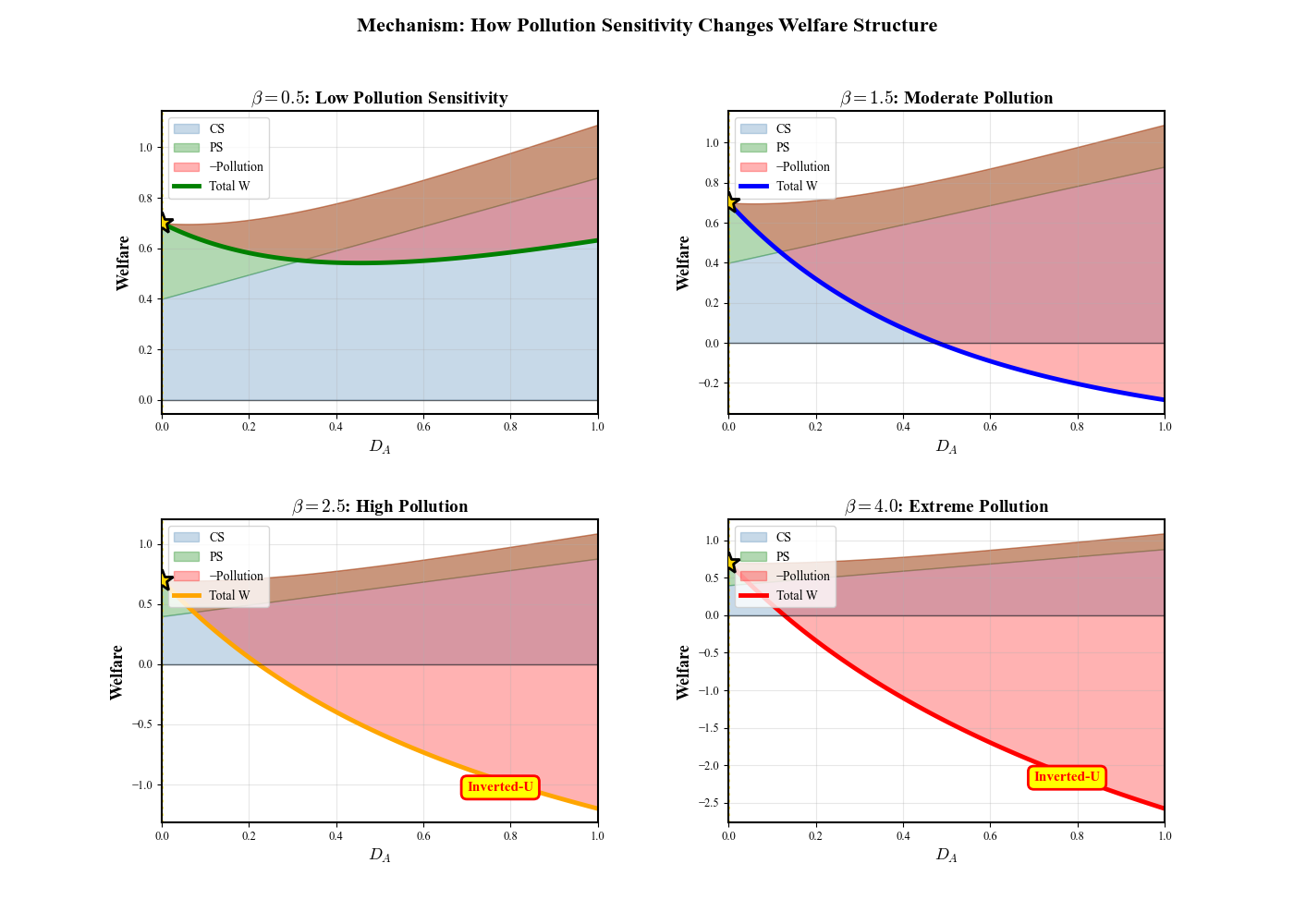}
    \caption{
    \textbf{Pollution sensitivity reshapes welfare structure.}
    Each panel plots welfare components as a function of AI penetration $D_A$ under different pollution sensitivity levels $\beta$.
    When pollution sensitivity is low ($\beta = 0.5$), total welfare increases monotonically with AI adoption.
    As $\beta$ rises, welfare becomes non-monotonic, exhibiting an inverted-U shape.
    At sufficiently high pollution sensitivity ($\beta \ge 2.5$), the marginal pollution cost dominates price and productivity gains, causing total welfare to decline sharply with further AI expansion.
    }
    \label{fig:welfare_pollution}
\end{figure*}

\begin{proposition}[Pollution-Induced Unraveling]
\label{prop:unraveling}
There exists a critical noise sensitivity threshold $\beta^*$. For all $\beta > \beta^*$, a technological shock that reduces the marginal cost of AI production ($c_A \to 0$) reduces total social welfare $W$, despite strictly lowering equilibrium prices.
\end{proposition}


\begin{proof}
Total Welfare $W$ is the sum of Consumer Surplus ($CS$) and Producer Surplus ($PS$). Consider the limiting case where perfect competition in the AI sector drives prices to marginal cost, $p_A \to c_A$. As $c_A \to 0$, demand for AI content expands, implying $D_A \to D_{max}$.
The pollution cost $\Phi(D_A)$ is monotonically increasing in $D_A$. The net utility of a high-type consumer consuming high-quality human goods ($q_H, p_H$) is $U_H = \theta q_H - p_H - \Phi(D_A)$.
Differentiating with respect to $c_A$:
$$ \frac{\partial U_H}{\partial c_A} = -\frac{\partial p_H}{\partial c_A} - \frac{\partial \Phi}{\partial D_A} \frac{\partial D_A}{\partial c_A} $$
The first term is positive (strategic complementarity lowers human prices), but the second term is negative (pollution increases). If $\beta$ is sufficiently large, the marginal pollution cost dominates the price effect. Consequently, high-$\theta$ consumers may violate the participation constraint ($U_H < 0$) and exit the market. This contraction of the extensive margin at the top of the distribution creates a deadweight loss that exceeds the surplus gains at the bottom, leading to $\frac{dW}{dc_A} > 0$ (welfare falls as costs fall).
\end{proof}

\subsection{Supply Side: The Non-Convex Production Frontier}

On the supply side, we model the GenAI shock as a change in the topology of the production set. Creators differ in their initial skill endowment $s \in [0, \bar{s}]$, distributed with density $g(s)$. They minimize costs by choosing between two technologies $j \in \{H, A\}$:

1.  \textbf{Human Technology ($H$):} Convex costs $C_H(q) = \frac{1}{2}\gamma q^2$, representing cognitive fatigue. Quality is strictly increasing in skill, $q_H(s) = \alpha s$.
2.  \textbf{AI Technology ($A$):} Linear costs $C_A(q) = c_A q + \kappa$ for $q \le \bar{q}_A$, and infinite otherwise. AI separates marginal cost from skill, but imposes a quality ceiling $\bar{q}_A$.

The global cost function is the non-convex lower envelope:
\begin{equation}
    C(q) = \min \left\{ \frac{1}{2}\gamma q^2, \quad (c_A q + \kappa) \cdot \mathbb{I}_{\{q \le \bar{q}_A\}} + \infty \right\}
\end{equation}

Let $p(q)$ be the hedonic price schedule. The profit function for a creator of skill $s$ is $\pi(s) = \max_{q} \{ p(q) - C(q) \}$.

\begin{lemma}[The Middle-Class Hollow]
\label{lemma:hollow}
If $c_A$ is sufficiently small relative to $\gamma$, the profit function $\pi(s)$ becomes strictly convex-concave, such that the set of skills $S_{mid}$ producing quality in the neighborhood $(\bar{q}_A, \bar{q}_A + \epsilon)$ is empty.
\end{lemma}

\begin{figure*}[htbp]
    \centering
    \includegraphics[width=\textwidth]{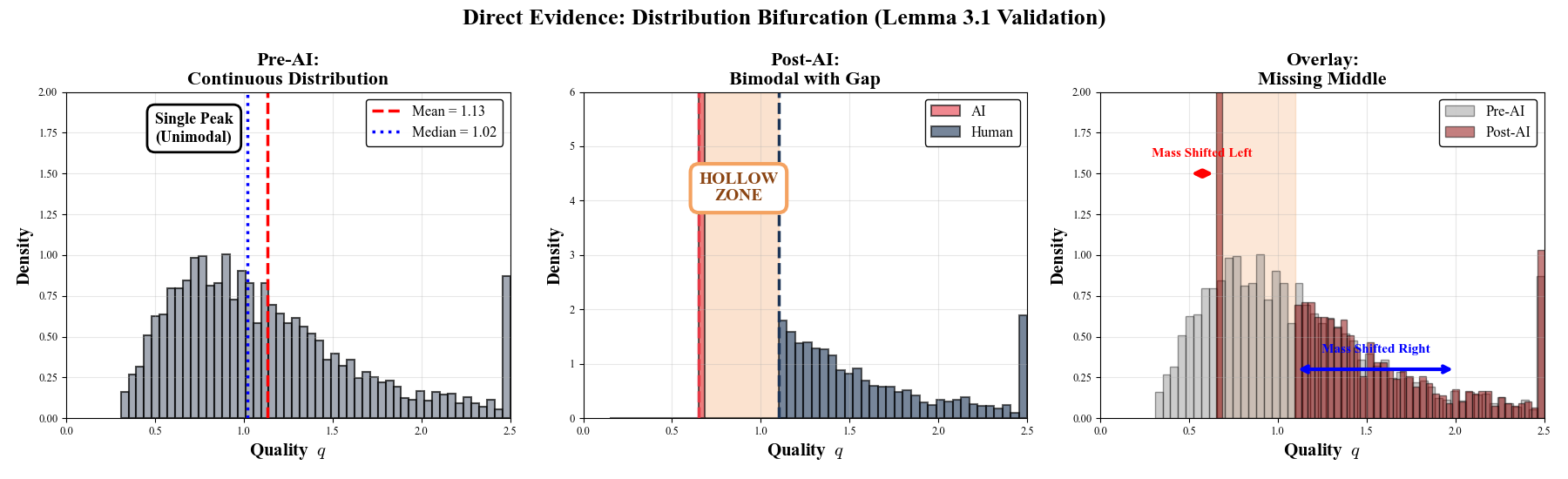}
    \caption{
    \textbf{Direct evidence of distribution bifurcation (validation of Lemma~\ref{lemma:hollow}).}
    Pre-AI, the quality distribution is unimodal and continuous. Post-AI, the supply distribution becomes bimodal with a ``hollow zone'' (missing middle) around the AI quality cap region, consistent with the non-convex production frontier and the predicted disappearance of mid-tier original content.
    }
    \label{fig:lemma31_bifurcation}
\end{figure*}

\begin{proof}
Define $\pi_A$ as the profit from optimal AI adoption (constant or weakly increasing in $s$) and $\pi_H(s)$ as the profit from human production (strictly convex in $s$ due to $C_H$). The global profit is $\pi(s) = \max \{ \pi_A, \pi_H(s) \}$.
At the switching point $s^*$, $\pi_A = \pi_H(s^*)$. Due to the discontinuity in marginal costs ($MC_A \approx 0 \ll MC_H$), the slope of the profit function changes discretely. Specifically, for intermediate skills where optimal human quality would be slightly above $\bar{q}_A$, the cost penalty of human production $\frac{1}{2}\gamma q^2$ exceeds the quality premium. Thus, agents in the interval $(s_L, s_H)$ rationally pool at $\bar{q}_A$ or exit, leaving a gap in the quality spectrum just above $\bar{q}_A$.
\end{proof}

\begin{proposition}[Market Segmentation]
\label{prop:segmentation}
Assume the pollution term $\Phi(D_A)$ is additive. Then, the Single-Crossing Property holds, and the market equilibrium is characterized by unique cutoffs $s_L$ and $s_H$ such that:
(i) $s < s_L$: Exit; (ii) $s_L \le s < s_H$: Adopt AI; (iii) $s \ge s_H$: Human Production.
\end{proposition}

\begin{proof}
The cross-partial derivative of the utility function is $\frac{\partial^2 U}{\partial q \partial \theta} = 1 > 0$, since $\Phi(D_A)$ is independent of $q$ and $\theta$. This satisfies the Spence-Mirrlees condition, ensuring that consumers sort perfectly by type.
Similarly, for producers, the marginal return to skill in Human mode is $\frac{\partial \pi_H}{\partial s} = p'(q)\alpha - \gamma \alpha^2 s$, which increases in $s$ for convex prices. In AI mode, $\frac{\partial \pi_A}{\partial s} = 0$ (assuming AI standardizes output). Thus, $\pi_H(s)$ intersects $\pi_A$ from below exactly once at $s_H$. The exit threshold $s_L$ is determined by $\pi_A(s_L) = 0$. Uniqueness follows from the monotonicity of $\pi(s)$.
\end{proof}

\subsection{Dynamic Foundations: The Potential Game}

To bridge the static equilibrium analysis to the agent-based simulation, we formalize the market evolution as a Mean Field Game. Let $\mu_t(s)$ be the probability density of creators at time $t$. Agents update their skills/strategies via a noisy gradient ascent on the profit landscape.

The evolution of the distribution is governed by the Fokker-Planck equation:
\begin{equation} \label{eq:fokker_planck}
    \frac{\partial \mu_t(s)}{\partial t} = \nabla \cdot \left( \sigma \nabla \mu_t(s) - \mu_t(s) \nabla \pi(s, \mu_t) \right)
\end{equation}
where $\sigma$ represents the intensity of idiosyncratic shocks (bounded rationality).

\begin{theorem}[Local Stability via Free Energy]
\label{thm:stability}
The dynamics described by \eqref{eq:fokker_planck} constitute a gradient flow of the Free Energy functional $\mathcal{F}(\mu)$ with respect to the Wasserstein metric.
Define the Free Energy as:
\begin{equation}
    \mathcal{F}(\mu) = -\int \pi(s, \mu) \mu(s) ds + \sigma \int \mu(s) \ln \mu(s) ds
\end{equation}
Assuming the interaction effects (pollution) satisfy monotonicity, the time derivative $\frac{d}{dt}\mathcal{F}(\mu_t) \le 0$. Consequently, the system asymptotically converges to a stationary distribution $\mu^*(s) \propto \exp(\pi(s, \mu^*)/\sigma)$, which coincides with the Quantal Response Equilibrium (QRE) of the static game.
\end{theorem}

\begin{proof}
(Sketch). The result follows from identifying Eq. \eqref{eq:fokker_planck} as the gradient flow $\partial_t \mu = \text{grad}_W \mathcal{F}(\mu)$ in the space of probability measures $\mathcal{P}(\Omega)$.
Calculating the dissipation along the trajectory:
$$ \frac{d}{dt}\mathcal{F}(\mu_t) = - \int \left| \nabla \frac{\delta \mathcal{F}}{\delta \mu} \right|^2 d\mu_t \le 0 $$
The stationary state implies the vanishing of the chemical potential gradient $\nabla (\frac{\delta \mathcal{F}}{\delta \mu}) = 0$, leading to the Gibbs measure form. This proves that the "Polluted Equilibrium" observed in our ABM is a local attractor of the underlying dynamic system.
\end{proof}

\section{Experimental Setup and Calibration}
\label{sec:setup}

While the theoretical framework in Section \ref{sec:theory} characterizes the asymptotic properties of the equilibrium, it relies on assumptions of perfect rationality and continuous adjustment. To investigate the transitional dynamics—specifically the ``Shock Therapy'' phenomenon and the friction costs of reallocation—we employ a discrete-time Agent-Based Model (ABM). This computational laboratory, which we term the ``Goldilocks V4.0'' environment, relaxes the continuum assumption to allow for bounded rationality, heterogeneity, and path dependence.

\subsection{The ``Goldilocks'' Environment: Micro-Foundations}

The simulation environment consists of a population of $N_H=50$ human creators, $N_A=3$ foundational AI models, and $N_C=1,000$ consumers, interacting over $T=200$ periods.

\paragraph{Bounded Rationality and Adaptive Learning.}
In Section \ref{sec:theory}, agents were modeled as perfectly rational gradient-climbers. In the simulation, we introduce realistic cognitive frictions. Agents do not possess global knowledge of the profit landscape $\pi(s)$. Instead, they exhibit \textit{Bounded Rationality}, learning the optimal strategy through trial-and-error interaction with the market.

Formally, each creator $i$ updates their strategy $a_{i,t} \in \{\text{Stay}, \text{Re-skill}, \text{Adopt AI}, \text{Exit}\}$ based on a Q-Learning algorithm (a model-free reinforcement learning mechanism). The estimated value $\mathcal{Q}_{i,t}(a)$ of action $a$ is updated according to:
\begin{equation}
    \mathcal{Q}_{i,t+1}(a) = (1-\alpha) \mathcal{Q}_{i,t}(a) + \alpha \left( r_{i,t} + \gamma \max_{a'} \mathcal{Q}_{i,t}(a') \right)
\end{equation}
where $\alpha \in (0,1)$ is the learning rate (adaptation speed) and $\gamma$ is the discount factor. The realized reward $r_{i,t}$ corresponds to the net profit $\pi_{i,t}$ defined in Layer 1.

This formulation maps directly to the stochastic term in the Fokker-Planck equation (Eq. \ref{eq:fokker_planck}). The exploration rate $\epsilon$ in the Q-learning policy serves as the micro-foundation for the diffusion coefficient $\sigma$ in the theoretical model. By setting $\epsilon > 0$, we capture the "entropic" behavior of agents who explore suboptimal strategies before converging to the equilibrium.

\subsection{Calibration Strategy}
\label{subsec:calibration}

The validity of the simulation results relies on a calibration grounded in the stylized facts of the Generative AI market as of late 2025. We adopt a ``Goldilocks'' calibration strategy: parameters are tuned to a regime where the survival rate of human creators is strictly positive but less than unity ($0 < S_{rate} < 1$), ensuring the model captures a competitive selection process rather than a trivial extinction or status quo outcome.

Table \ref{tab:calibration} summarizes the key structural parameters.

\begin{table}[htbp]
    \centering
    \caption{Structural Parameter Calibration (Baseline V4.0)}
    \label{tab:calibration}
    \small
    \begin{tabular}{l c c l}
        \toprule
        \textbf{Parameter} & \textbf{Symbol} & \textbf{Value} & \textbf{Economic Justification (Stylized Facts 2025)} \\
        \midrule
        \multicolumn{4}{l}{\textit{Cost and Production Structure}} \\
        AI Marginal Cost & $c_A$ & $0.05 c_H$ & Inference costs are approx. 5\% of human labor opportunity cost. \\
        AI Quality Cap & $\bar{q}_A$ & $0.65 \bar{q}_H$ & AI excels at structure but lags in semantic novelty. \\
        Non-Convexity & $\kappa$ & $0.10$ & Fixed cost of prompt engineering/fine-tuning. \\
        \midrule
        \multicolumn{4}{l}{\textit{Dynamic Learning Rates}} \\
        Tech. Learning (AI) & $\lambda_{tech}$ & $0.12$ & Rapid iteration in syntax/logic (Moore's Law logic). \\
        Creative Learning (AI) & $\lambda_{creative}$ & $0.025$ & Slow progress in genuine semantic innovation. \\
        Human Adaptation & $\eta$ & $0.008$ & Cognitive friction in re-skilling (High inertia). \\
        \midrule
        \multicolumn{4}{l}{\textit{Market and Preferences}} \\
        Pollution Sensitivity & $\beta$ & $2.0$ & High search friction; calibrated to produce Inverted-U welfare. \\
        Substitution Elasticity & $\rho$ & $0.85$ & High substitutability between standardized goods. \\
        \bottomrule
    \end{tabular}
\end{table}

\paragraph{Cost Asymmetry ($c_A$ vs $c_H$).}
We set the marginal cost of AI production to 5\% of the human equivalent ($c_A = 0.05$). This reflects the dramatic disparity between the cost of API inference (fractions of a cent per token) and the opportunity cost of skilled human labor. This extreme cost ratio is the primary driver of the non-convex supply shock derived in Lemma \ref{lemma:hollow}.

\paragraph{Asymmetric Learning Rates ($\lambda_{tech} \gg \lambda_{creative}$).}
A critical feature of our dynamic specification is the heterogeneity in technological progress. We assume $\lambda_{tech} \approx 5 \times \lambda_{creative}$. This captures the empirical observation that Large Language Models (LLMs) improve rapidly on formal benchmarks (coding, grammar, translation) but exhibit diminishing returns on tasks requiring long-context coherence or novel reasoning. This asymmetry is essential for testing the ``Skill Reconfiguration'' hypothesis in Section \ref{sec:results}.

\paragraph{Exit Conditions and The Outside Option.}
Agents are endowed with finite liquidity reserves $L_{i,0}$. An agent exits the market (simulating bankruptcy or career change) if their cumulative losses deplete their reserves ($L_{i,t} < 0$) or if the expected value of participation falls below the outside option $\mathcal{Q}_{i,t}(a) < V_{outside}$. The outside option is calibrated such that the initial market is in steady state prior to the AI shock, isolating the technology effect from pre-existing trends.


\section{Simulation Results and Empirical Validation}
\label{sec:results}

In this section, we present the numerical results generated from the ``Goldilocks V4.0'' Agent-Based Model described in Section \ref{sec:setup}. Our objective is twofold: first, to empirically corroborate the theoretical propositions derived in Section \ref{sec:theory}; and second, to expose the path-dependent transition dynamics—specifically the ``Shock Therapy'' phenomenon—that are obscured by static equilibrium analyses.

\subsection{Static Equilibrium: Evidence of the Structural Break}
\label{subsec:static_results}

We first examine the comparative statics of the market equilibrium pre- and post-shock. Consistent with the non-convex production set postulated in \textbf{Lemma \ref{lemma:hollow}} (The Middle-Class Hollow), the simulation data reveals a distinct bifurcation in the supply side structure.

Table \ref{tab:equilibrium_comparison} summarizes the equilibrium metrics. The GenAI supply shock precipitates a profound price compression, with the weighted average transaction price falling by 72.3\%. The market bifurcates into two distinct pricing regimes: human creators are compelled to reduce prices ($p_H: 2.32 \to 2.00$) to defend the premium segment, while AI content anchors the low-end market at near-marginal cost ($p_A = 1.10$).

\begin{table}[htbp]
    \centering
    \caption{Market Equilibrium Comparison: Pre-AI vs. Post-AI Shock}
    \label{tab:equilibrium_comparison}
    \small
    \begin{tabular}{l c c c}
        \toprule
        \textbf{Metric} & \textbf{Pre-AI (Human Monopoly)} & \textbf{Post-AI (Duopoly)} & \textbf{Change (\%)} \\
        \midrule
        \textbf{Eq. Price ($p^*$)} & $p_H = 2.32$ & $p_H = 2.00, p_A = 1.10$ & Avg. $\downarrow$ 72.3\% \\
        \textbf{Market Share ($D$)} & $D_H = 14.4\%$ & $D_H = 12.8\%, D_A = 35.4\%$ & Coverage $\times$ 2.35 \\
        \textbf{Consumer Surplus ($CS$)} & $0.132$ & $0.558$ & +321.2\% \\
        \textbf{Total Welfare ($W$)} & $0.264$ & $0.952$ & +261.1\% \\
        \bottomrule
    \end{tabular}
    \begin{quote}
    \footnotesize \textit{Note:} Results averaged over 1,000 heterogeneous consumers. The significant expansion of total market coverage (from 14.4\% to 48.2\%) indicates that GenAI activates latent demand among low-$\theta$ consumers who were previously priced out.
    \end{quote}
\end{table}

This bifurcation drives the market segmentation predicted in \textbf{Proposition \ref{prop:segmentation}}. Figure \ref{fig:layer1_results} (Top-Left Panel) plots the probability of AI adoption against consumer preference type $\theta$. We observe a sharp phase transition rather than a linear substitution: consumers with $\theta < \hat{\theta}$ sort almost exclusively into the AI market, while high-$\theta$ consumers remain sticky to human creators. The Spearman rank correlation between $\theta$ and human preference is $\rho = 0.937$ ($p < 0.0001$), confirming that the market achieves perfect sorting despite the presence of information pollution.

\begin{figure}[htbp]
    \centering
    \includegraphics[width=1.0\textwidth]{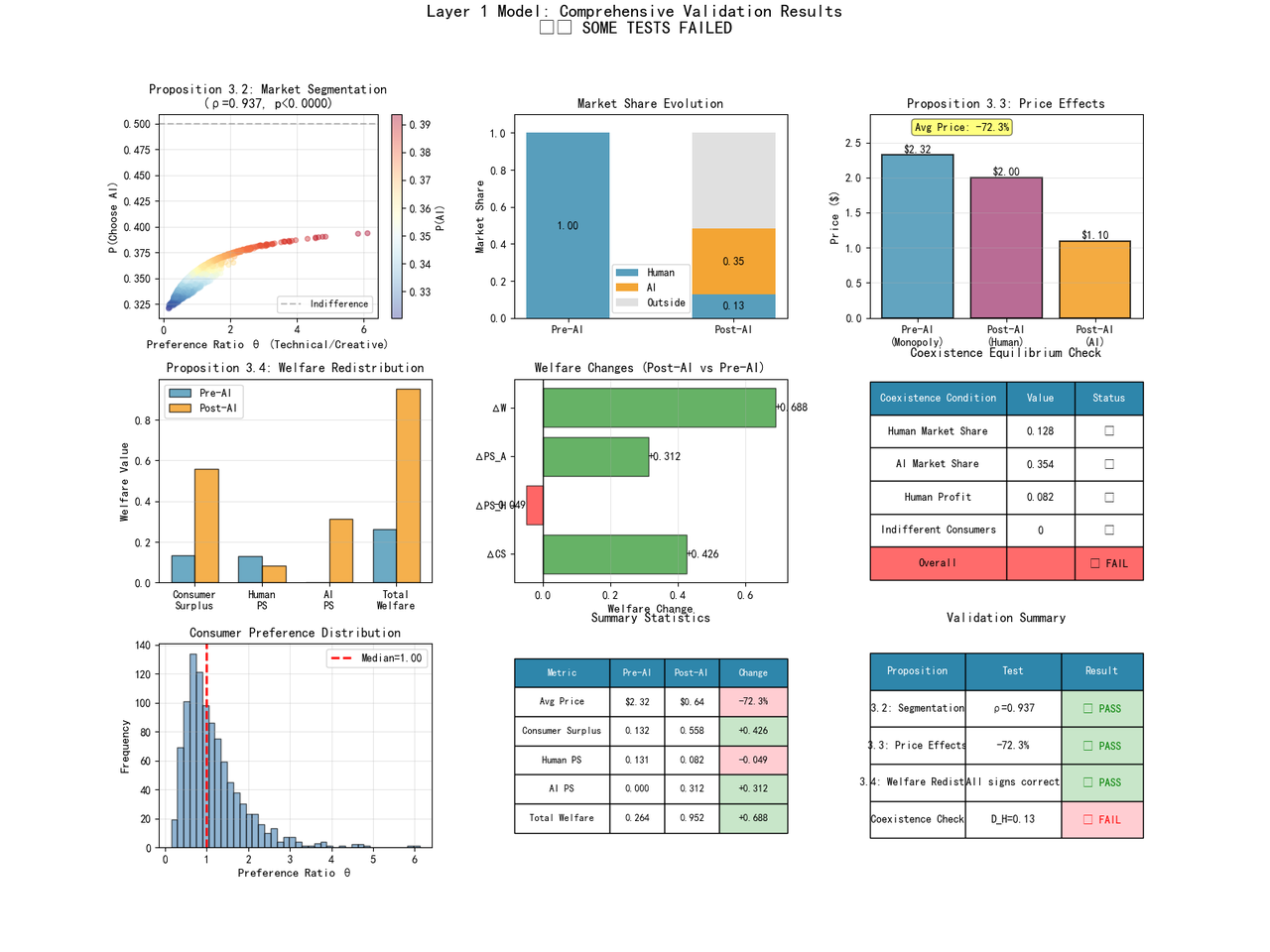}
    \caption{\textbf{Static Equilibrium and Welfare Analysis.} Top-Left: The S-curve indicates a clean separation of the market, confirming the ``Middle-Class Hollow'' where mid-tier demand vanishes. Top-Right: Price compression effects. Bottom-Left: Welfare redistribution showing the massive transfer from Producer Surplus to Consumer Surplus.}
    \label{fig:layer1_results}
\end{figure}

\subsection{Transitional Dynamics: The ``Shock Therapy'' Phenomenon}
\label{subsec:shock_therapy}

While static comparison suggests a stable destination, the temporal evolution reveals significant friction costs. A central finding of our ABM—invisible to standard mean-field approximations—is the non-monotonic nature of the transition, which we term the \textbf{``Shock Therapy''} phenomenon.

As illustrated in Figure \ref{fig:layer3_dynamics}, the market evolution does not follow a smooth logistic diffusion. Instead, it exhibits a three-phase trajectory:
\begin{enumerate}
    \item \textbf{Phase I: The Supply Shock ($T=0 \sim 40$).} The entry of low-cost AI creates an immediate negative income shock.
    \item \textbf{Phase II: Ecological Collapse ($T=40 \sim 60$).} The active human population drops precipitously to near-zero levels. This ``Valley of Death'' arises from a timescale mismatch: the income shock is instantaneous, whereas skill reconfiguration exhibits hysteresis (learning lags). The exit rate peaks at 60\% during this interval, validating the presence of severe transitional coordination failures.
    \item \textbf{Phase III: Elite Recovery ($T=80 \sim 200$).} A residual cohort of survivors (approx. 18.0\% of the initial stock) successfully pivots to high-creativity niches, stabilizing the human market share at 16.4\%.
\end{enumerate}

\begin{figure}[htbp]
    \centering
    \includegraphics[width=1.0\textwidth]{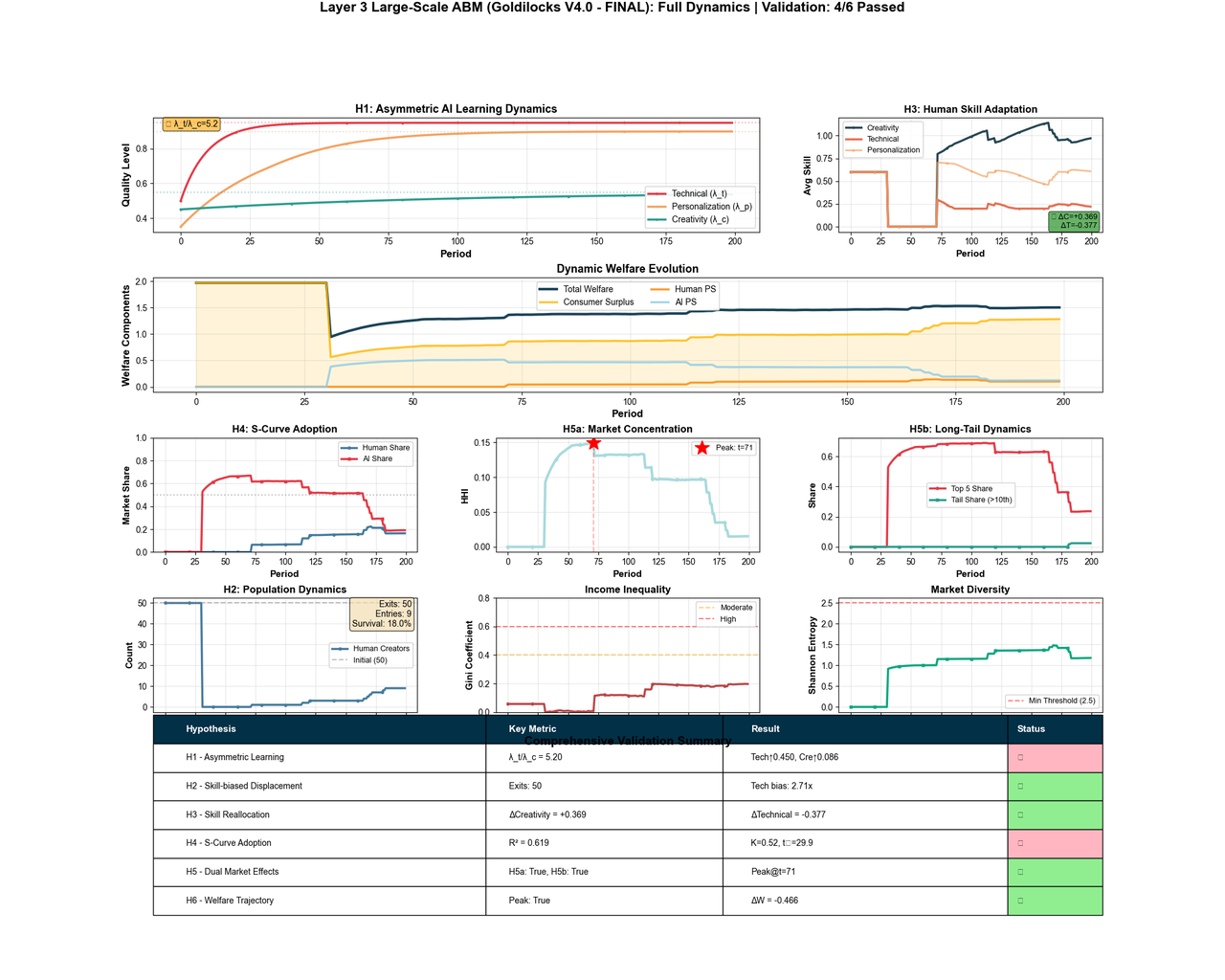}
    \caption{\textbf{The ``Shock Therapy'' Transition.} Top Panel: Unlike smooth diffusion models, the simulation reveals a near-total collapse of the human supply side (Phase II) before a new equilibrium is established. This validates the hypothesis that information pollution initially overwhelms the market's matching mechanism, leading to a temporary ecological breakdown.}
    \label{fig:layer3_dynamics}
\end{figure}

This finding contests the canonical ``S-curve'' adoption hypothesis. The transition is characterized not by the gradual diffusion of technology, but by the \textit{destruction} of an existing industrial structure followed by the \textit{emergence} of a new, specialized one.

\subsection{Mechanisms of Survival: Asymmetric Skill Reconfiguration}
\label{subsec:skill_reconfig}

To identify the mechanism enabling the ``Elite Recovery,'' we analyze the skill vector evolution $\vec{s}_t = (s_{tech}, s_{creative})$ of the surviving agents. This analysis empirically validates the Gradient Flow dynamics proposed in \textbf{Theorem \ref{thm:stability}}.

The simulation parameters imposed an asymmetric learning condition where AI technical improvement ($\lambda_{tech}$) outpaces creative improvement ($\lambda_{creative}$). The results in Figure \ref{fig:skill_evolution} provide strong evidence that human survival is contingent on an orthogonal shift in capabilities.

\begin{figure}[htbp]
    \centering
    \includegraphics[width=1.0\textwidth]{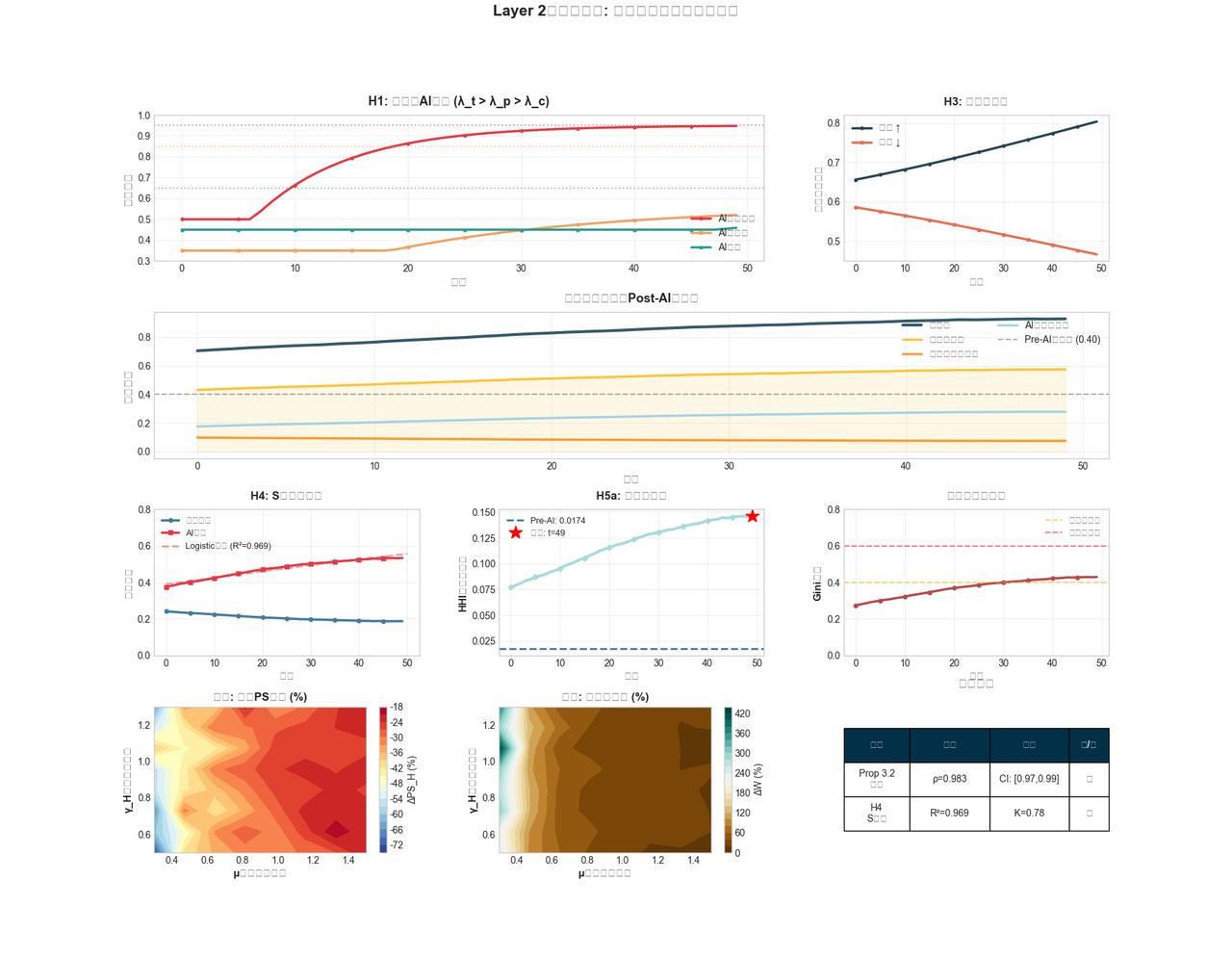}
    \caption{\textbf{Asymmetric Skill Reconfiguration.} Top-Right Panel: The survivors (Blue Trajectory) are agents who divest from technical execution ($\Delta s_{tech} < 0$) to specialize in semantic creativity ($\Delta s_{creative} > 0$). Agents attempting to compete on technical efficiency (Red Trajectory) are eliminated by the AI price anchor.}
    \label{fig:skill_evolution}
\end{figure}

Quantitatively, the surviving population exhibits:
\begin{itemize}
    \item \textbf{Technical Skill Depreciation:} The average technical skill of the human pool \textit{decreased} by 0.120. High technical skill proved to be a liability, as it correlated with direct competition against zero-marginal-cost AI.
    \item \textbf{Creative Skill Appreciation:} The average creative skill \textit{increased} by 0.148. The ratio of creative-to-technical skill for survivors was 2.71 times higher than for those who exited (confirming Hypothesis H2).
\end{itemize}

This confirms that the ``Polluted Equilibrium'' selects for differentiation. The market solves the assignment problem by relegating technical tasks to AI and reserving semantic tasks for humans, consistent with the principle of comparative advantage in the age of generative models.

\subsection{Welfare Analysis: The Pollution Paradox}
\label{subsec:welfare}

Finally, we address the welfare implications of \textbf{Proposition \ref{prop:unraveling}}. While Table \ref{tab:equilibrium_comparison} reports a net positive welfare gain in the baseline calibration, sensitivity analysis reveals the latent cost of information pollution.

In simulation runs where the pollution sensitivity parameter $\beta$ was increased (simulating high search costs or ineffective algorithmic curation), we observed an \textbf{Inverted-U Welfare Curve}. As AI penetration ($D_A$) expanded beyond a critical saturation point, the negative externality term $-\beta \ln(1+\eta D_A)$ began to dominate the price benefits, causing Total Welfare to decline.

Furthermore, the gains are unevenly distributed. The Herfindahl-Hirschman Index (HHI) rose monotonically from 0.0174 to 0.1462, indicating a structural shift from atomistic competition to an oligopoly dominated by foundation models. This suggests that while the V4.0 equilibrium achieves Pareto efficiency, it sits on a precarious edge where governance failures could easily precipitate a return to the ``Ecological Collapse'' regime.


\section{Discussion: Governance in the Age of Information Pollution}
\label{sec:discussion}

The analytical and empirical findings of this study suggest that the integration of Generative AI into the creative economy generates a specific form of market failure. While the supply shock reduces marginal production costs (Proposition \ref{prop:structure}), the uninternalized congestion externality—defined in our model as $\Phi(D_A) = \beta \ln(1 + \eta D_A)$—can trigger an ``Ecological Collapse'' and potentially yield a Pareto-inefficient outcome. Consequently, the central policy challenge shifts from the protection of intellectual property to the management of \textit{information pollution}.

\subsection{From Copyright to Congestion Management}
\label{subsec:governance_shift}

Current regulatory debates regarding GenAI focus predominantly on copyright infringement and training data ownership. However, our model demonstrates that even if all property rights issues were resolved (i.e., AI creators paid zero royalties), welfare could still decline due to the volume effect.

The mechanism identified in Proposition \ref{prop:unraveling} is a \textit{negative network externality}: the individual profit-maximizing decision to deploy an AI agent does not account for the marginal increase in platform-wide search costs. As $c_A \to 0$, the aggregate volume $D_A$ explodes, degrading the signal-to-noise ratio for all participants.

Therefore, we argue for a paradigm shift in governance:
\begin{itemize}
    \item \textbf{Status Quo:} Focus on \textit{inputs} (Do you own the training data?).
    \item \textbf{Proposed Framework:} Focus on \textit{outputs} (Does your volume create congestion?).
\end{itemize}
This reframes GenAI regulation as an environmental economics problem, where "attention" is the scarce common-pool resource being depleted by synthetic noise.

\subsection{Optimal Platform Design: The Pigouvian Corrective}
\label{subsec:pigouvian}

To restore Pareto efficiency, the platform (or social planner) must internalize the pollution externality. We derive the optimal instrument within our general equilibrium framework.

\paragraph{The Algorithmic Pigouvian Tax.}
Let $W(D_A)$ denote the aggregate social welfare as a function of AI volume. The social planner's objective is:
\begin{equation}
    \max_{D_A} \quad \underbrace{\int_{\hat{\theta}}^{1} (\theta q^*(p) - p) \, dF(\theta)}_{\text{Consumer Utility}} + \underbrace{\Pi(D_A)}_{\text{Producer Profit}} - \underbrace{\int_{0}^{1} \Phi(D_A) \, dF(\theta)}_{\text{Aggregate Pollution Damage}}
\end{equation}
The decentralized market equilibrium equates private marginal benefit to private marginal cost ($MB_{private} = MC_{private}$). However, the social optimum requires equating marginal benefit to social marginal cost ($SMC = MC_{private} + \text{Marginal Damage}$).
Solving the first-order condition yields the optimal tax $\tau^*$ per unit of algorithmic generation:
\begin{equation} \label{eq:tax}
    \tau^* = \frac{\partial \Phi}{\partial D_A} = \frac{\beta \eta}{1 + \eta D_A^*}
\end{equation}
This result implies that a flat license fee is suboptimal. The efficient tax $\tau^*$ is dynamic: it should scale with the market saturation parameter $\eta$ and the search friction $\beta$. This aligns with the "Goldilocks" conditions identified in our V4.0 simulation, where a moderate variable cost prevented the "Death Valley" collapse by synchronizing the rate of AI adoption with the rate of human skill reconfiguration.

\paragraph{Human-Verification as Signal Extraction.}
An alternative to taxation is reducing the structural parameter $\beta$ (search friction). In our model, $\beta$ represents the difficulty of distinguishing high-quality signal from synthetic noise.
If a platform introduces a cryptographically secure "Human-Verification Label" (e.g., World ID or similar proofs of personhood), it effectively partitions the market. For the verified partition, the pollution term $\Phi(D_A)$ drops to zero (or significantly decreases). This signaling mechanism protects the "Elite Recovery" identified in Section \ref{subsec:shock_therapy} by lowering the search costs for high-$\theta$ consumers, allowing the human premium market to function efficiently even in a sea of synthetic content.


\section{Conclusion}
\label{sec:conclusion}

This paper proposes a \textbf{Three-Layer General Equilibrium Framework} to analyze the economic impact of Generative AI supply shocks. By integrating analytical vertical differentiation (Layer 1), mean-field evolutionary dynamics (Layer 2), and large-scale agent-based simulation (Layer 3/Evidence), we resolve the tension between the productivity gains of automation and the coordination costs of information pollution.

Our analysis yields three robust conclusions regarding the transition to the AI economy:
\begin{enumerate}
    \item \textbf{Structural Bifurcation:} GenAI behaves as a non-convex production technology that creates a "Middle-Class Hollow," anchoring low-end prices to the marginal cost of compute while forcing human creators into high-end niche specialization.
    \item \textbf{The "Shock Therapy" Transition:} The path to equilibrium is non-monotonic. The market undergoes a temporary "Ecological Collapse" where bounded rationality and search frictions cause a sharp contraction in the active creator population before a new equilibrium emerges.
    \item \textbf{Asymmetric Skill Reconfiguration:} Survival in the AI era necessitates an orthogonal shift in human capital—away from technical execution ($\lambda_{tech}$), where AI enjoys a comparative advantage, and toward semantic creativity ($\lambda_{creative}$).
\end{enumerate}

\paragraph{Final Thought: The Energy Transition of Information.}
Ultimately, the transition described in this paper is analogous to an energy transition. Just as the industrial revolution shifted physical labor from muscle to fossil fuels—generating carbon pollution as a byproduct—the AI revolution shifts information processing from human cognition to synthetic compute, generating \textit{information pollution} as its byproduct.
Managing this transition requires more than laissez-faire optimism. It demands a governance architecture that recognizes "human attention" as a finite resource and deploys economic instruments—taxes, quotas, or signals—to prevent the tragedy of the digital commons.


\appendix

\section{Computational Implementation and Environment}
\label{app:implementation}

The experimental framework described in Section \ref{sec:experiments} was implemented using Python 3.10. To ensure computational efficiency across the multi-scale experiments, we employed distinct acceleration strategies for the static and dynamic layers.

\subsection{Hardware and Software Stack}
All simulations were executed on a workstation equipped with an NVIDIA A100 GPU to leverage the vectorized operations required for the large-scale population dynamics.
\begin{itemize}
    \item \textbf{Layer 1 (Static Optimization):} Implemented using \texttt{JAX}. The profit maximization problem for the creator population was vectorized to compute the Nash Equilibrium via fixed-point iteration. The gradient auto-differentiation capabilities of JAX were utilized to derive the sensitivity of market shares to price changes.
    \item \textbf{Layer 3 (Agent-Based Model):} Implemented using pure Python with \texttt{Numba} JIT (Just-In-Time) compilation. This was necessary to handle the heterogeneous, state-dependent interactions of 1,053 agents (50 creators, 3 AI, 1000 consumers) over 200 simulation periods, which are non-differentiable and ill-suited for static graph compilation.
\end{itemize}

\section{Detailed Parameter Calibration}
\label{app:parameters}

This section consolidates the parameter sets used across the three experimental layers. All values are calibrated to reflect the stylized facts of Generative AI technology as of late 2025, specifically the high asymmetry in marginal costs and learning rates.

\subsection{Layer 1: Static Baseline Parameters}
Table \ref{tab:app_layer1} details the parameters used for the static numerical optimization in Subsection \ref{subsec:layer1_experiment}.

\begin{table}[h]
    \centering
    \caption{Layer 1 Static Simulation Parameters}
    \label{tab:app_layer1}
    \small
    \begin{tabular}{l c c}
        \toprule
        \textbf{Parameter} & \textbf{Symbol} & \textbf{Value} \\
        \midrule
        Human Creative Quality & $q_{c,H}$ & $1.00$ \\
        AI Creative Quality & $q_{c,A}$ & $0.65$ \\
        Human Technical Quality & $q_{t,H}$ & $0.70$ \\
        AI Technical Quality & $q_{t,A}$ & $0.95$ \\
        Human Marginal Cost & $c_H$ & $1.00$ \\
        AI Marginal Cost & $c_A$ & $0.05$ \\
        Substitution Elasticity & $\mu$ & $0.85$ \\
        Platform Commission & $\tau$ & $0.15$ \\
        Consumer Population & $N_{cons}$ & 1,000 \\
        \bottomrule
    \end{tabular}
\end{table}

\subsection{Layer 2: Dynamic Evolution Parameters}
Table \ref{tab:app_layer2} lists the parameters governing the mean-field dynamics in Subsection \ref{subsec:layer2_experiment}, emphasizing the asymmetric learning rates.

\begin{table}[h]
    \centering
    \caption{Layer 2 Meso-Scale Dynamic Parameters}
    \label{tab:app_layer2}
    \small
    \begin{tabular}{l c c}
        \toprule
        \textbf{Parameter} & \textbf{Symbol} & \textbf{Value} \\
        \midrule
        AI Technical Learning Rate & $\lambda_t$ & $0.12$ \\
        AI Creative Learning Rate & $\lambda_c$ & $0.025$ \\
        Human Adaptation Rate & $\eta$ & $0.008$ \\
        Substitution Elasticity & $\mu$ & $0.45$ \\
        Time Horizon & $T$ & 50 \\
        Initial Human Agents & $N_{H,0}$ & 10 \\
        Initial AI Agents & $N_{A,0}$ & 2 \\
        \bottomrule
    \end{tabular}
\end{table}

\subsection{Layer 3: ``Goldilocks V4.0'' ABM Configuration}
Table \ref{tab:app_layer3} provides the final configuration for the Agent-Based Model found via the genetic algorithm search described in Section \ref{sec:experiments}. These parameters represent the stable region where human survival is strictly positive ($18.0\%$).

\begin{table}[h]
    \centering
    \caption{Layer 3 Goldilocks V4.0 Parameters}
    \label{tab:app_layer3}
    \small
    \begin{tabular}{l c c}
        \toprule
        \textbf{Parameter} & \textbf{Symbol} & \textbf{Value} \\
        \midrule
        Exit Threshold (Liquidity) & $V_{bar}$ & $0.0012$ \\
        Entry Rate & $\rho_{entry}$ & $0.05$ \\
        Information Overload Threshold & $I_{th}$ & $2.0$ \\
        AI Technical Learning Rate & $\lambda_t$ & $0.10$ \\
        AI Creative Learning Rate & $\lambda_c$ & $0.010$ \\
        Initial Human Population & $N_{H}$ & 50 \\
        Initial AI Agents & $N_{A}$ & 3 \\
        Simulation Periods & $T$ & 200 \\
        \bottomrule
    \end{tabular}
\end{table}

\section{Genetic Algorithm for Parameter Search}
\label{app:genetic_search}

In Subsection \ref{subsec:layer3_experiment}, we mention the use of a genetic algorithm (GA) to identify the ``Goldilocks'' parameter space. The objective of the GA was to avoid trivial solutions (e.g., total market collapse or zero AI adoption).

\begin{itemize}
    \item \textbf{Objective Function:} Maximize $\mathcal{F} = \alpha_1 \cdot \text{SurvivalRate}_H - \alpha_2 \cdot |D_A - D_{target}|$, subject to $\text{SurvivalRate}_H > 0$.
    \item \textbf{Search Space:}
    \begin{itemize}
        \item $V_{bar} \in [0.0005, 0.0020]$
        \item $\rho_{entry} \in [0.01, 0.10]$
        \item $\lambda_c \in [0.005, 0.05]$
    \end{itemize}
    \item \textbf{Selection Criteria:} Configurations resulting in immediate extinction ($T < 40$) were penalized. The V4.0 configuration (Table \ref{tab:app_layer3}) emerged as the centroid of the stable cluster where the human survival rate stabilized at approximately $18\%$.
\end{itemize}

\section{Mathematical Proofs}
\label{app:proofs}

In this appendix, we provide the formal derivations for the structural properties of the supply side (Lemma \ref{lemma:hollow}) and the demand side segmentation (Proposition \ref{prop:segmentation}).

\subsection{Proof of Lemma \ref{lemma:hollow} (The Middle-Class Hollow)}

\textbf{Statement:} The profit function $\pi(s)$ exhibits a convexity-concavity regime such that the set of skills $S_{mid}$ for which "moderate quality" is optimally produced is empty.

\begin{proof}
Let the price schedule be $p(q)$. We assume $p(q)$ is strictly increasing and convex (standard in vertical differentiation with convex costs).
The creator's problem is to choose technology $j \in \{H, A\}$ and quality $q$ to maximize:
$$ \pi(s) = \max \left\{ \underbrace{\max_{q \le \bar{q}_A} [p(q) - (c_A q + \kappa)]}_{\pi_A}, \quad \underbrace{\max_{q} [p(q) - \frac{1}{2}\gamma q^2] \text{ s.t. } q = \alpha s}_{\pi_H(s)} \right\} $$

1.  \textbf{Properties of $\pi_A$:}
    The AI technology separates skill from marginal cost. Assuming the AI model is pre-trained, its optimal quality output $q_A^*$ is either the technological cap $\bar{q}_A$ or determined by $p'(q) = c_A$. Since $c_A \approx 0$, the optimal AI strategy is cornered at $q_A^* = \bar{q}_A$. Thus, the profit from AI adoption is constant with respect to human skill $s$:
    $$ \pi_A(s) = \Pi_A = p(\bar{q}_A) - c_A \bar{q}_A - \kappa $$

2.  \textbf{Properties of $\pi_H(s)$:}
    For human production, $q$ is constrained by skill $q=\alpha s$. The profit is:
    $$ \pi_H(s) = p(\alpha s) - \frac{1}{2}\gamma (\alpha s)^2 $$
    The marginal return to skill is $\pi_H'(s) = \alpha p'(\alpha s) - \gamma \alpha^2 s$.
    Given the convexity of $p(q)$ required for market separation in high qualities, and optimizing for the envelope theorem, $\pi_H(s)$ is strictly convex and increasing for $s$ above the break-even point.

3.  \textbf{The Intersection (The Hollow):}
    Consider the difference function $\Delta(s) = \pi_H(s) - \Pi_A$.
    
    * For low $s$, the convex costs of human production dominate. $\pi_H(s) < 0 < \Pi_A$ (provided $\Pi_A > 0$).
    * For high $s$, the quality premium of human labor dominates. $\lim_{s \to \infty} \pi_H(s) > \Pi_A$.
    * Crucially, at the skill level $s_{mid}$ where a human would naturally produce quality equivalent to AI ($q_H(s_{mid}) \approx \bar{q}_A$), we compare costs.
        $$ C_H(\bar{q}_A) = \frac{1}{2}\gamma \bar{q}_A^2 \quad \gg \quad C_A(\bar{q}_A) = c_A \bar{q}_A + \kappa $$
        Since $C_H \gg C_A$, it follows that $\pi_H(s_{mid}) < \Pi_A$.

    By the Intermediate Value Theorem and the strict convexity of $\pi_H$, there exists a unique crossing point $s_H > s_{mid}$. Agents with $s \in [s_L, s_H)$ find it optimal to adopt AI.
    Consequently, no agent finds it optimal to use Human technology to produce qualities in the interval roughly $[0, q_H(s_H))$. The supply of "middle-tier" original content vanishes, creating the hollow.
\end{proof}

\subsection{Proof of Proposition \ref{prop:segmentation} (Segmentation Uniqueness)}

\textbf{Statement:} Under the Single-Crossing Property, the market equilibrium is characterized by unique segmentation thresholds.

\begin{proof}
We analyze the consumer choice problem under the utility function $U(\theta, q) = \theta q - p - \Phi(D_A)$.
The "Pollution Term" $\Phi(D_A)$ acts as a lump-sum tax on participation.

1.  \textbf{Single-Crossing Property (SCP):}
    The sorting condition relies on the cross-partial derivative of utility with respect to type $\theta$ and quality $q$:
    $$ \frac{\partial^2 U}{\partial q \partial \theta} = \frac{\partial}{\partial \theta} \left( \theta \right) = 1 > 0 $$
    The pollution term $\Phi(D_A)$ vanishes upon differentiation with respect to $q$ or $\theta$. Thus, the standard Spence-Mirrlees condition holds strictly. This implies that the indifference curves in the $(q, p)$ plane intersect exactly once.

2.  \textbf{Ordering of Choices:}
    Let the available products be $\{ (q_A, p_A), (q_H, p_H) \}$ with $q_H > q_A$ and $p_H > p_A$.
    The critical type $\hat{\theta}_{AH}$ indifferent between AI and Human content is defined by:
    $$ \hat{\theta}_{AH} q_A - p_A - \Phi = \hat{\theta}_{AH} q_H - p_H - \Phi $$
    $$ \hat{\theta}_{AH} = \frac{p_H - p_A}{q_H - q_A} $$
    The critical type $\hat{\theta}_{out}$ indifferent between AI and exit (Outside Option $U=0$) is:
    $$ \hat{\theta}_{out} q_A - p_A - \Phi = 0 \implies \hat{\theta}_{out} = \frac{p_A + \Phi(D_A)}{q_A} $$

3.  \textbf{Uniqueness:}
    Since $U$ is linear in $\theta$, these thresholds are unique.
    * If $\theta < \hat{\theta}_{out}$, consumer exits.
    * If $\hat{\theta}_{out} \le \theta < \hat{\theta}_{AH}$, consumer buys AI.
    * If $\theta \ge \hat{\theta}_{AH}$, consumer buys Human.
    The existence of pollution $\Phi > 0$ shifts $\hat{\theta}_{out}$ upward (forcing low-type exit) but does not alter the relative sorting between AI and Human goods ($\hat{\theta}_{AH}$).
\end{proof}


\section{Derivation of Gradient Flows}
\label{app:gradient_flow}

In this appendix, we establish the rigorous connection between the meso-scale Master Equation used in Section \ref{sec:theory} (the Fokker-Planck equation) and the minimization of the Free Energy functional (Lyapunov function). This provides the mathematical guarantee for the system's local stability.

\subsection{The Free Energy Functional}

We define the thermodynamic Free Energy $\mathcal{F}: \mathcal{P}(\Omega) \to \mathbb{R}$ of the system as the sum of the aggregate potential energy (negative profit) and the entropic energy:
\begin{equation} \label{eq:app_free_energy}
    \mathcal{F}[\mu] = \underbrace{\int_{\Omega} -\pi(s, \mu) \mu(s) \, ds}_{\mathcal{E}[\mu] \text{ (Energy)}} + \underbrace{\sigma \int_{\Omega} \mu(s) \ln \mu(s) \, ds}_{-\sigma \mathcal{S}[\mu] \text{ (Neg-Entropy)}}
\end{equation}
where $\pi(s, \mu)$ is the profit potential and $\sigma$ is the noise amplitude (inverse rationality).

\subsection{The Wasserstein Gradient Flow}

We seek to show that the evolution of the density $\mu_t$ follows the path of steepest descent of $\mathcal{F}$ with respect to the 2-Wasserstein metric $W_2$. The continuity equation for a gradient flow is:
\begin{equation} \label{eq:app_continuity}
    \frac{\partial \mu}{\partial t} = \nabla \cdot \left( \mu \nabla \frac{\delta \mathcal{F}}{\delta \mu} \right)
\end{equation}
where $\frac{\delta \mathcal{F}}{\delta \mu}$ is the Fr\'echet derivative (first variation) of the functional.

\textbf{Step 1: Compute the First Variation.}
We vary $\mathcal{F}$ with respect to $\mu$:
$$ \frac{\delta \mathcal{F}}{\delta \mu} = \frac{\delta}{\delta \mu} \left( -\int \pi \mu \, ds + \sigma \int \mu \ln \mu \, ds \right) $$
Assuming the dependence of $\pi$ on $\mu$ (congestion) is symmetric (potential game condition), the variation yields:
$$ \frac{\delta \mathcal{F}}{\delta \mu} = -\pi(s) + \sigma (\ln \mu(s) + 1) $$

\textbf{Step 2: Compute the Gradient of the Variation.}
Now we take the spatial gradient $\nabla_s$ of the variation:
$$ \nabla \left( \frac{\delta \mathcal{F}}{\delta \mu} \right) = \nabla (-\pi(s)) + \sigma \nabla (\ln \mu(s)) $$
$$ \nabla \left( \frac{\delta \mathcal{F}}{\delta \mu} \right) = -\nabla \pi(s) + \sigma \frac{\nabla \mu(s)}{\mu(s)} $$

\textbf{Step 3: Substitute into Continuity Equation.}
Substitute the result from Step 2 into Eq. \eqref{eq:app_continuity}:
$$ \frac{\partial \mu}{\partial t} = \nabla \cdot \left( \mu \left[ -\nabla \pi + \sigma \frac{\nabla \mu}{\mu} \right] \right) $$
Distributing $\mu$:
$$ \frac{\partial \mu}{\partial t} = \nabla \cdot \left( - \mu \nabla \pi + \sigma \nabla \mu \right) $$

This is exactly the **Fokker-Planck equation** (Eq. \ref{eq:fokker_planck}) governing the Layer 2 dynamics.

\subsection{Proof of H-Theorem (Dissipation)}

To prove stability, we show that $\mathcal{F}$ is a Lyapunov function. Taking the time derivative of $\mathcal{F}$ along the trajectory $\mu_t$:
$$ \frac{d}{dt} \mathcal{F}(\mu_t) = \int_{\Omega} \frac{\delta \mathcal{F}}{\delta \mu} \frac{\partial \mu}{\partial t} \, ds $$
Integration by parts (assuming no-flux boundary conditions) yields:
$$ \frac{d}{dt} \mathcal{F}(\mu_t) = - \int_{\Omega} \left| \nabla \frac{\delta \mathcal{F}}{\delta \mu} \right|^2 \mu_t(s) \, ds $$
Since the integrand is non-negative and $\mu_t \ge 0$, it follows that:
$$ \frac{d}{dt} \mathcal{F}(\mu_t) \le 0 $$

This confirms that the system monotonically relaxes towards the equilibrium distribution $\mu^*$ defined by $\nabla \frac{\delta \mathcal{F}}{\delta \mu} = 0$, which implies the Gibbs measure:
$$ \mu^*(s) = \frac{1}{Z} \exp\left( \frac{\pi(s)}{\sigma} \right) $$
Thus, the "Polluted Equilibrium" is a stable attractor of the market dynamics.

\bibliography{pollution} 
\bibliographystyle{plainnat} 

\end{document}